\documentclass{llncs}
\usepackage[papersize={8.5in,11in},top=1.48in,left=1.48in,right=1.48in,bottom=1.48in]{geometry}
\usepackage{graphicx}
\usepackage{latexsym}
\usepackage{times}
\usepackage{setspace}
\PassOptionsToPackage{boxed,section}{algorithm}
\usepackage{algorithm}
\usepackage{algorithmic}
\usepackage{enumerate}
\usepackage{amsmath}			
\usepackage{amssymb}			
\usepackage{url}
\usepackage{setspace}
\usepackage{tikz}
\usetikzlibrary{arrows}
\usetikzlibrary{shapes}

\long\def\comment#1{}

\renewcommand\belowcaptionskip{-5pt}

\begin{document}

\mainmatter
\title{On Partial Vertex Cover on Bipartite Graphs and Trees\thanks{This research has been supported in part by the National Science Foundation through Award CNS-0849735 and CCF-0827397, and the Air Force of Scientific Research through Award FA9550-12-1-0199.}}
\author{
 Bugra Caskurlu and K. Subramani}
\institute{
 LDCSEE, \\ West Virginia University, \\
 Morgantown, WV \\
\email{caskurlu@gmail.com,ksmani@csee.wvu.edu}
}

\pagestyle{empty}

\maketitle

\begin{abstract}
It is well-known that the Vertex Cover problem is in {\bf P} on bipartite graphs, however; the computational complexity of the Partial Vertex Cover problem on bipartite graphs is open. In this paper, we first show that the Partial Vertex Cover problem is {\bf NP-hard} on bipartite graphs. We then identify an interesting special case of bipartite graphs, for which the Partial Vertex Cover problem can be solved in polynomial-time. We also show that the set of acyclic bipartite graphs, i.e., forests, and the set of bipartite graph where the degree of each vertex is at most $3$ fall into that special case. Therefore, we prove that the Partial Vertex Cover problem is in {\bf P} on trees, and it is also in {\bf P} on the set of bipartite graphs where the degree of each vertex is at most $3$.
\end{abstract}

\section{Introduction}
\label{sec:intro}

Covering problems arise often in practice. A mobile phone service provider should ensure that its base stations cover the signals transmitted from the phones of its customers. A chain market such as Walmart should ensure that it has a store close to its customers. The applications of the covering problems are not limited to corporations to sell a service to customers. The Air Force on a no-fly zone mission or border patrol officers trying to secure borders are to solve some form of a covering problem.

In many real life situations the corporations or the government is constrained in the resources it can allocate for the covering mission. The constraints may be hard constraints such as a government agency to operate within its approved budget, or profit dictated soft constraints such as a mobile phone service provider may decide not to cover a rural area since the revenues will not match the covering costs. Therefore, the goal in many real life situation can be cast as covering the domain as much as possible for a given fixed amount of resources to be allocated.

There is merit in studying partial covering problems both due to their wide applicability in a large range of applications and their theoretical importance as being natural generalizations of classical covering problems. In this paper, we study the Partial Vertex Cover (PVC) problem on bipartite graphs and trees, and use the notation PVCB and PVCT to denote them. Though it is well-known that the Vertex Cover problem is polynomial-time solvable on bipartite graphs, the computational complexity of both the PVCB and PVCT problems are open. However, there are provable good approximation algorithms for these problems, and the best approximation algorithm in the literature has an approximation ratio of $(\frac{4}{3}+ \epsilon)$ \cite{KPS11}.

\bigskip

The principal contributions of this paper are as follows:
 \begin{enumerate}[(i)]
  \item The PVCB problem is {\bf NP-hard}.
  \item The Partial Vertex Cover problem is polynomial-time solvable on the set of bipartite graphs that has the Marginally Nonincreasing Coverage (MNC) property, which is defined in Section \ref{sec:MNCProperty}.
  \item Trees have the MNC property, and therefore, the PVCT problem is in {\bf P}.
  \item The set of bipartite graphs where the degree of each vertex is at most $3$ has the MNC property, and therefore, the Partial Vertex Cover problem is polynomial-time solvable on them.
 \end{enumerate}

\bigskip

The rest of this paper is organized as follows: Section \ref{sec:sop} presents a formal definition of the PVC problem. We present a concrete application to motivate the PVCB problem in Section \ref{sec:motiv}. The related work is presented in Section \ref{sec:relwork}. The computational complexity of the PVCB problem is established in Section \ref{sec:comp}. In Section \ref{sec:MNCProperty}, we show some interesting special cases for which the PVCB problem is polynomial-time solvable. We conclude and point out several research directions in Section \ref{sec:conc}.

\section{Statement of Problems}
\label{sec:sop}
In the classical Vertex Cover (VC) problem, we are given an undirected graph ${\bf G= \langle V, E \rangle}$, where ${\bf V}$ is the vertex set with $|{\bf V}| = n$, ${\bf E}$ is the edge set with $|{\bf E}| = m$. The goal is to find a minimum cardinality subset ${\bf V'} \subset {\bf V}$, such that for every edge $e = (i,j) \in {\bf E}$, either $i \in {\bf V'}$ or $j \in {\bf V'}$.

\bigskip

In the Partial Vertex Cover (PVC) problem, we are given an integer $t$, and an undirected graph ${\bf G= \langle V, E \rangle}$, where ${\bf V}$ is the vertex set with $|{\bf V}| = n$, ${\bf E}$ is the edge set with $|{\bf E}| = m$. The goal is to find a minimum cardinality subset ${\bf V'} \subset {\bf V}$ such that ${\bf V'}$ {\em covers} at least $t$ edges, i.e., for at least $t$ edges $(i,j) \in {\bf E}$, either $i \in {\bf V'}$ or $j \in {\bf V'}$. It is trivial to observe that the PVC problem is a generalization of the VC problem, since the PVC problem subsumes the VC problem for $t = m$.

\section{Motivation}
\label{sec:motiv}

A concrete application of the Partial Vertex Cover problem on bipartite graphs on risk analysis is given in \cite{CGBS13}. In \cite{CGBS13}, the risk of a computational system is modeled as a flow between the first and last partitions in a tripartite graph, where the vertices of the three partitions represent threats to the system, vulnerabilities of the system, and the assets of the system as shown in Figure \ref{fig:motiv}.

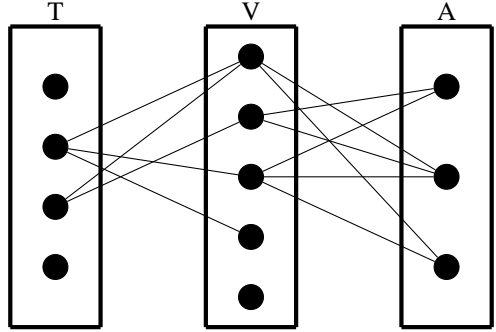
\begin{figure}[ht]
\begin{center}
 \begin{tikzpicture}[scale=0.4]

  \node at (-6.5, 5.5) (l13) {T};
  \node at (0, 5.5) (l14) {V};
  \node at (6.5, 5.5) (l15) {A};

  \draw[ultra thick] (-8, 5)--(-5, 5);
  \draw[ultra thick] (-5, 5)--(-5, -5);
  \draw[ultra thick] (-5, -5)--(-8, -5);
  \draw[ultra thick] (-8, -5)--(-8, 5);

  \draw[ultra thick] (-1.5, 5)--(1.5, 5);
  \draw[ultra thick] (1.5, 5)--(1.5, -5);
  \draw[ultra thick] (1.5, -5)--(-1.5, -5);
  \draw[ultra thick] (-1.5, -5)--(-1.5, 5);

  \draw[ultra thick] (5, 5)--(8, 5);
  \draw[ultra thick] (8, 5)--(8, -5);
  \draw[ultra thick] (8, -5)--(5, -5);
  \draw[ultra thick] (5, -5)--(5, 5);

 \node[circle,fill=black,draw] at (-6.5,3) (n1) {};
 \node[circle,fill=black,draw] at (-6.5,1) (n2) {};
 \node[circle,fill=black,draw] at (-6.5,-1) (n3) {};
 \node[circle,fill=black,draw] at (-6.5,-3) (n4) {};
 \node[circle,fill=black,draw] at (0,4) (n5) {};
 \node[circle,fill=black,draw] at (0,2) (n6) {};
 \node[circle,fill=black,draw] at (0,0) (n7) {};
 \node[circle,fill=black,draw] at (0,-2) (n8) {};
 \node[circle,fill=black,draw] at (0,-4) (n9) {};
 \node[circle,fill=black,draw] at (6.5,3) (n10) {};
 \node[circle,fill=black,draw] at (6.5,0) (n11) {};
 \node[circle,fill=black,draw] at (6.5,-3) (n12) {};

  \draw (n2)--(n5);
  \draw (n2)--(n7);
  \draw (n2)--(n8);
  \draw (n3)--(n5);
  \draw (n3)--(n6);
  \draw (n5)--(n11);
  \draw (n5)--(n12);
  \draw (n6)--(n10);
  \draw (n6)--(n11);
  \draw (n7)--(n10);
  \draw (n7)--(n11);
  \draw (n7)--(n12);
 \end{tikzpicture}

\end{center}
\label{fig:motiv}
\caption{Risk in a computational system can be modeled in terms of its
  constituent components. The threats, weaknesses (corresponding to
  specific vulnerabilities), and assets form three disjoint sets, named as $T, V$, and $A$ respectively. An edge between vertices represents a
  contribution to the system risk. The system's risk is the total flow
  between the first and third sets.}
\end{figure}

In the risk management model given in \cite{CGBS13}, the goal is to reduce the risk in the system (flow between the first and last partitions) below a predefined risk threshold level by either restricting the permissions of the users, or encapsulating the system assets. These two strategies correspond to deleting minimum number of vertices from the second and the third partitions of the tripartite graph so that the flow between the first and the third partitions are reduced below the predefined threshold level. The equivalence of this risk management system and the Partial Vertex Cover problem on bipartite graphs is established in \cite{CGBS13}.

\section{Related Work}
\label{sec:relwork}

 The VC problem is one of the classical {\bf NP-complete} problems listed by Karp \cite{Kar72}. There are several polynomial-time approximation algorithms for the VC problem within a factor of $2$, and the best-known approximation algorithm for the VC problem has an approximation factor of $2 - \theta \left(\frac{1}{\sqrt{\log n}}\right)$ \cite{Kar09}. The VC problem is known to be {\bf APX-complete} \cite{PY91}. Moreover, it cannot be approximated to within a factor of $1.3606$ unless {\bf P = NP} \cite{DS05}, and not within any constant factor smaller than $2$, unless the \textit{unique games conjecture} is false \cite{KR08}.

\bigskip

Since the PVC problem subsumes the VC problem for $t = m$, all the hardness results given above for the VC problem directly apply to the PVC problem. The PVC problem and the partial variants of similar graph problems have been extensively studied for more than a decade \cite{Bla03}, \cite{KMR07}, \cite{KLR08}, \cite{KMRR06}, \cite{BFMR10}. In particular, there is a $O(n \cdot \log n + m)$-time primal-dual $2$-approximation algorithm\cite{M09}, a combinatorial $2$-approximation algorithm \cite{BFMR10}, and several $(2 - o(1))$-approximation algorithms \cite{Bar01}, \cite{BshB98}, \cite{GKS04}, \cite{Hoch98}.

\bigskip

Though the VC problem and the PVC problem has almost matching approximation ratios and inapproximability results, the PVC problem is in some sense more difficult than the VC problem. For instance, the PVC problem is W[1]-complete while the VC problem is fixed parameter tractable \cite{Guo05}.

\section{Computational Complexity of the PVCB Problem}
\label{sec:comp}

This section is devoted to prove Theorem \ref{thm:complexity}, i.e., the PVCB problem is {\bf NP-hard}.

\begin{theorem}
\label{thm:complexity}
The Partial Vertex Cover problem is {\bf NP-hard} on bipartite graphs.
\end{theorem}

\begin{proof}
We will prove Theorem \ref{thm:complexity} by giving a Karp reduction from the CLIQUE problem. In the CLIQUE problem, we are given an undirected graph ${\bf G'= \langle V', E'\rangle}$, and an integer $k$, and the goal is to find whether there exists a complete subgraph of ${\bf G'}$ with $k$ vertices. Assume we are given an arbitrary undirected graph ${\bf G'= \langle V', E'\rangle}$, where $n'$ and $m'$ denote $|{\bf V'}|$ and $|{\bf E'}|$ respectively, and an integer $k$. We construct a corresponding bipartite graph ${\bf G= \langle V_1 \cup V_2, E \rangle}$ as explained below.

\bigskip

For every vertex $v'_i \in {\bf V'}$, ${\bf G}$ has a corresponding vertex $v_i \in {\bf V_2}$. For every edge $e' \in {\bf E'}$, there is a corresponding \textit{edge block} in ${\bf G}$. The term \textit{edge block} refers to two vertices and an edge in between. So, for each edge $e' \in {\bf E'}$, ${\bf G}$ has two corresponding vertices $e_1 \in {\bf V_2}$ and $e_2 \in {\bf V_1}$ and the edge $(e_1, e_2)$. In order to capture the incidence matrix of ${\bf G'}$, for each edge $e' = (v'_i, v'_j)$ of ${\bf G'}$, ${\bf G}$ has $2$ additional edges $(v_i, e_1)$ and $(v_j, e_1)$. Let $n$ and $m$ denote the number of vertices and edges of ${\bf G}$, respectively. Notice that the bipartite graph ${\bf G}$ has $n' + 2 \cdot m'$ vertices and $3 \cdot m$ edges. More precisely, we have $n = n' + 2 \cdot m'$, and $m = 3 \cdot m'$. We use the term \textit{left vertex} of an edge block for the vertex of the edge block that belongs to ${\bf V_1}$. The other vertex of the edge block that belongs to ${\bf V_2}$, is referred to as the \textit{right vertex} of the edge block throughout the paper.

\bigskip

In Figure \ref{fig:PVCB}, we are given a simple undirected graph ${\bf G'}$ on the left that consists of 2 vertices $v'_1$ and $v'_2$ and an edge $e' = (v'_1, v'_2)$ in between. The figure has the corresponding bipartite graph ${\bf G}$ on the right. The vertices $v_1$ and $v_2$ of ${\bf G}$ correspond to 2 vertices of ${\bf G'}$. The 2 vertices $e_1$ and $e_2$ of ${\bf G}$ and the edge $e$ in between is the corresponding edge block of the edge $e'$ of ${\bf G'}$. The 2 edges $(v_1, e_1)$ and $(v_2, e_1)$ of ${\bf G}$ capture the incidence matrix of ${\bf G'}$.

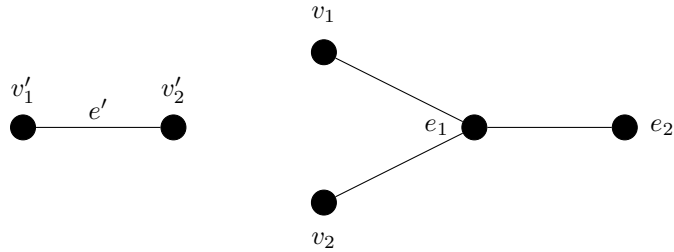
\begin{figure}[htbp]
\begin{center}
\begin{tikzpicture}
\node at (0, .5) (h) {An undirected graph ${\bf G'}$ and corresponding bipartite graph ${\bf G}$};

\node [fill=black, circle, draw] at (1, -1) (v1) {};
\node at (1, -.5) (v13) {$v_1$};
\node [fill=black, circle, draw] at (3, -2) (v4) {};
\node at (2.5, -2) (e1) {$e_1$};
\node [fill=black, circle, draw] at (5, -2) (v5) {};
\node at (5.5, -2) (e2) {$e_2$};
\node [fill=black, circle, draw] at (1, -3) (v6) {};
\node at (1, -3.5) (v23) {$v_2$};

\draw (v4)--(v5);
\draw (v1)--(v4);
\draw (v6)--(v4);

\node [fill=black, circle, draw] at (-3, -2) (v9) {};
\node at (-3, -1.5) (v_1) {$v'_1$};
\node [fill=black, circle, draw] at (-1, -2) (v10) {};
\node at (-1, -1.5) (v_2) {$v'_2$};

\node at (-2, -1.75) (e) {$e'$};

\draw (v9)--(v10);
\end{tikzpicture}
\end{center}

\caption{Construction of the corresponding bipartite graph ${\bf G}$ for a given undirected graph ${\bf G'}$.}
\label{fig:PVCB}
\end{figure}

We will prove Theorem \ref{thm:complexity} by showing that solving the CLIQUE problem on ${\bf G'}$ reduces to solving the PVCB problem on ${\bf G}$ with $t = m - \frac{k \cdot (k-1)}{2}$. In the rest of the paper, without loss of generality, we will assume that $m' > \frac{k \cdot (k - 1)}{2}$, and $k \geq 5$. Notice that these assumptions do not violate the soundness of the proof since the CLIQUE problem is still {\bf NP-hard} after these assumptions on the input. We precisely will show the following: There is a complete subgraph of $k$ vertices on ${\bf G'}$ if and only if there is a subset of $k + m' - \frac{k \cdot (k-1)}{2}$ vertices ${\bf V''}$ of ${\bf G}$ such that the number of edges that are covered by ${\bf V''}$ is at least $m - \frac{k \cdot (k-1)}{2}$.

\bigskip

Without loss of generality we can assume that ${\bf V''}$ does not contain the right vertex of any of the edge blocks of ${\bf G}$, since the right vertex of an edge block is incident to only one edge and that edge can be covered by the left vertex of the edge block as well. Therefore, without loss of generality, we can assume that all the vertices of ${\bf V''}$ are the vertices that correspond to the vertices of ${\bf G'}$, or the left vertices of the edge blocks. Since the number of edges that are to be covered by ${\bf V''}$ is at least $m - \frac{k \cdot (k-1)}{2}$, the number of edges that are not covered by ${\bf V''}$ is at most $\frac{k \cdot (k-1)}{2}$. Therefore, ${\bf V''}$ has to contain the left vertices of at least $m' - \frac{k \cdot (k-1)}{2}$ edge blocks. We will complete the proof of Theorem \ref{thm:complexity} by proving Lemma \ref{lem:yes}, which maps the yes instances of the CLIQUE problem to the yes instances of the PVCB problem, and Lemma \ref{lem:no}, which maps the no instances of the CLIQUE problem to the no instances of the PVCB problem.

\begin{lemma}
\label{lem:yes}
If there exists a complete subgraph of $k$ vertices on ${\bf G'}$, then there exists a subset ${\bf V''}$ of vertices of ${\bf G}$ such that $|{\bf V''}| = k + m' - \frac{k \cdot (k-1)}{2}$, and ${\bf V''}$ covers at least $m - \frac{k \cdot (k-1)}{2}$ edges of ${\bf G}$.
\end{lemma}

\begin{proof}
Assume that ${\bf G'}$ has a complete subgraph of $k$ vertices and let ${\bf V''}$ be composed of the following $k + m' - \frac{k \cdot (k-1)}{2}$ vertices of ${\bf G}$. For every vertex of the complete subgraph of ${\bf G'}$, let the corresponding vertex of ${\bf G}$ be in ${\bf V''}$. Notice that there are exactly $k$ such vertices. The complete subgraph of this $k$ vertices has $\frac{k \cdot (k-1)}{2}$ edges in ${\bf G'}$. Therefore, there are $m' - \frac{k \cdot (k-1)}{2}$ edges of ${\bf G'}$ that are \textit{not} in the complete subgraph of $k$ vertices in ${\bf G'}$. For each of these $m' - \frac{k \cdot (k-1)}{2}$ edges of ${\bf G'}$, let the left vertex of the corresponding edge block in ${\bf G}$ be contained in ${\bf V''}$. Notice that there are exactly $m' - \frac{k \cdot (k-1)}{2}$ such vertices in ${\bf V''}$. So, $|{\bf V''}| = k + m' - \frac{k \cdot (k-1)}{2}$ as desired.

\bigskip

All we need to prove is that ${\bf V''}$ covers at least $m - \frac{k \cdot (k-1)}{2}$ edges of ${\bf G}$. Let us first consider the edges of ${\bf G}$ that capture the incidence relation of the edges of ${\bf G'}$. Recall that for every edge $e' = (v'_i, v'_j)$ of ${\bf G'}$, there are $2$ edges in ${\bf G}$ to capture the incidence relation of $e'$, namely $(v_i, e_1)$ and $(v_j, e_2)$. So, in total there are $2 \cdot m'$ such edges in ${\bf G}$. The $k \cdot (k-1)$ edges of ${\bf G}$ that capture the incidence relation of the $\frac{k \cdot (k-1)}{2}$ edges of the complete subgraph of ${\bf G'}$ are covered by the $k$ vertices of ${\bf V''}$ that correspond to the $k$ vertices of ${\bf G'}$. The remaining $2 \cdot m' - k \cdot (k-1)$ edges of ${\bf G}$, that capture the incidence relation of the $m' - \frac{k \cdot (k-1)}{2}$ edges of ${\bf G'}$ that are not part of the complete subgraph, are covered by the left vertices of the $m' - \frac{k \cdot (k-1)}{2}$ edge blocks in ${\bf V''}$. Therefore, all $2 \cdot m'$ edges of ${\bf G}$ that capture the incidence relation of the $m'$ edges of ${\bf G'}$ are covered by ${\bf V''}$.

\bigskip

Recall that there are $m'$ additional edges in ${\bf G}$. These edges are the edges of the $m'$ edge blocks. The left vertices of the $m' - \frac{k \cdot (k-1)}{2}$ edge blocks that are contained in ${\bf V''}$ cover $m' - \frac{k \cdot (k-1)}{2}$ of those edges. Therefore, there are only $\frac{k \cdot (k-1)}{2}$ edges of ${\bf G}$ that are not covered by ${\bf V''}$. So, ${\bf V''}$ covers $m - \frac{k \cdot (k-1)}{2}$ edges as stated by Lemma \ref{lem:yes}.
\end{proof}

\begin{lemma}
\label{lem:no}
If ${\bf G'}$ does not have a complete subgraph of $k$ vertices, then no subset ${\bf V''}$ of vertices of ${\bf G}$ such that $|{\bf V''}| = k + m' - \frac{k \cdot (k-1)}{2}$ covers at least $m - \frac{k \cdot (k-1)}{2}$ edges of ${\bf G}$.
\end{lemma}

\begin{proof}
Assume ${\bf G'}$ does not have a complete subgraph of $k$ vertices. For the purpose of contradiction, assume that there is a subset ${\bf V''}$ of vertices of ${\bf G}$ such that $|{\bf V''}| = k + m' - \frac{k \cdot (k-1)}{2}$, and ${\bf V''}$ covers at least $m - \frac{k \cdot (k-1)}{2}$ edges of ${\bf G}$.

\bigskip

Since ${\bf V''}$ covers at least $m - \frac{k \cdot (k-1)}{2}$ edges of ${\bf G}$ and $m' > \frac{k \cdot (k-1)}{2}$, ${\bf V''}$ covers at least $m' - \frac{k \cdot (k-1)}{2}$ edges of the edge blocks. Therefore, ${\bf V''}$ contains the left vertices of at least $m' - \frac{k \cdot (k-1)}{2}$ edge blocks. Since $|{\bf V''}| = k + m' - \frac{k \cdot (k-1)}{2}$, ${\bf V''}$ contains at most $k$ vertices of ${\bf G}$ that correspond to the vertices of ${\bf G'}$.

\bigskip

First consider the case where ${\bf V''}$ contains exactly $k$ vertices of ${\bf G}$ that correspond to the vertices of ${\bf G'}$, and exactly $m' - \frac{k \cdot (k-1)}{2}$ left vertices of edge blocks. Since there are only $m' - \frac{k \cdot (k-1)}{2}$ left vertices of edge blocks, ${\bf V''}$ does not cover $\frac{k \cdot (k-1)}{2}$ edges of the edge blocks. Since ${\bf V''}$ covers at least $m - \frac{k \cdot (k-1)}{2}$ in total, ${\bf V''}$ covers all the edges of ${\bf G}$ that capture the incidence relation of all the edges of ${\bf G'}$. Since ${\bf G'}$ does not have a complete subgraph of $k$ vertices, the $k$ vertices of ${\bf G}$ that correspond to the vertices of ${\bf G'}$ cover both of the edges that capture the incidence relation of $\frac{k \cdot (k-1)}{2} - \alpha$ edges of ${\bf G'}$ for some $1 \leq \alpha < \frac{k \cdot (k-1)}{2}$. Since ${\bf V''}$ covers all $2 \cdot m'$ edges of ${\bf G}$ that capture the incidence relation of the edges of ${\bf G'}$, ${\bf V''}$ contains the left vertices of all the edge blocks of ${\bf G}$ that correspond to the $m' - \frac{k \cdot (k-1)}{2} + \alpha$ edges of ${\bf G'}$. This is a contradiction since we assumed that ${\bf V''}$ contains exactly $m - \frac{k \cdot (k-1)}{2}$ left vertices of edge blocks.

\bigskip

Therefore, ${\bf V''}$ contains exactly $k - l$ vertices of ${\bf G}$ that correspond to some $k - l$ vertices of ${\bf G'}$, and exactly $m' - \frac{k \cdot (k-1)}{2} + l$ left vertices of edge blocks for some $0 < l < k$. Recall that the incidence relation of each edge of ${\bf G'}$ is captured by $2$ edges in ${\bf G}$. Notice that the subgraph formed by this $k - l$ vertices of ${\bf G'}$ contains at most $\frac{(k - l) \cdot (k - l - 1)}{2}$ edges of ${\bf G'}$.  Therefore, the corresponding $k - l$ vertices of ${\bf V''}$ covers both of the incidence edges of exactly $\frac{(k - l) \cdot (k - l - 1)}{2}$ edges of ${\bf G'}$. In other words, at least one incidence edge of $m' - \frac{(k - l) \cdot (k - l - 1)}{2}$ edges of ${\bf G'}$ is not covered the $k - l$ vertices of ${\bf G}$ that correspond to some $k - l$ vertices of ${\bf V''}$. Since we already have $\frac{k \cdot (k-1)}{2} - l$ edges of edge blocks left uncovered, the left vertices of the edge blocks in ${\bf V''}$ has to cover an incidence edge for at least $m' - \frac{(k - l) \cdot (k - l - 1)}{2} - l$ edges of ${\bf G'}$. This is not possible since the left vertex of each edge block in ${\bf G}$ covers the corresponding incidence edges of exactly one edge of ${\bf G'}$, and $m' - \frac{(k - l) \cdot (k - l - 1)}{2} - l > m' - \frac{k \cdot (k-1)}{2} + l$ for $k \geq 5$.
\end{proof}

\end{proof}

\section{Marginally Nonincreasing Coverage Property}
\label{sec:MNCProperty}

Let ${\bf B}$ denote an arbitrary bipartite graph. We define $OPT_B (k)$ as the maximum number of edges of ${\bf B}$ that can be covered by a subset of $k$ vertices. We say that the Partial Vertex Cover problem has MNC property on a bipartite graph ${\bf B}$ if $OPT_B (k+2) - OPT_B (k+1) \leq OPT_B (k+1) - OPT_B (k)$ for all $k \geq 0$.

\begin{theorem}
\label{thm:PifMNPholds}
Let ${\bf B}$ be a bipartite graph such that the MNC property holds on ${\bf B}$. Then the PVCB problem can be solved in polynomial-time on ${\bf B}$.
\end{theorem}

The rest of the section is devoted to proving Theorem \ref{thm:PifMNPholds}.

Our proof is constructive, i.e., we show that there exists a polynomial-time algorithm that solves the PVCB problem exactly if the MNC property holds. The algorithm we give is indeed the $(\frac{4}{3} + \epsilon)$-approximation algorithm given by \cite{KPS11}. We analyze their algorithm under the assumption that the MNC property holds and prove that the algorithm returns the exact solution.

\bigskip

We first start with the following IP formulation of the PVCB problem:

\[\begin{array}{lll}
\mbox{minimize} & \sum_{i \in {\bf V}} x_i & \mbox{Objective Function}\\
\mbox{subject to} & x_i + x_j + z_e \geq 1 \qquad \forall e = (i,j) \in {\bf E} & \mbox{Constraints}\\
& \sum_{e \in E} z_e \leq (m - t)&\\
& x_i, z_e \in \{0,1\} \qquad \forall i \in {\bf V}, e \in {\bf E} & \mbox{Variables}\\
\end{array} \]

In this IP formulation, we have a variable $x_i$ for every vertex $i \in {\bf V}$. If the vertex $i$ is in the partial cover, then $x_i = 1$. Otherwise, $x_i = 0$. For every edge $e \in {\bf E}$, we have a variable $z_e$ to denote whether $z_e$ is uncovered or not. $z_e = 1$ if $e$ is \textit{uncovered}. $z_e = 0$ if $e$ is covered.

\bigskip

If we take the Lagrangian relaxation of the IP and remove the constant term from the objective, we will obtain the following Lagrangian IP:

\[\begin{array}{lll}
\mbox{minimize} & \sum_{i \in {\bf V}} w_i \cdot x_i + \lambda \sum_{e \in {\bf E}} z_e & \mbox{Objective Function}\\
\mbox{subject to} & x_i + x_j + z_e \geq 1 \qquad \forall e = (i,j) \in {\bf E} & \mbox{Constraints}\\
& x_i, z_e \in \{0,1\} \qquad \forall i \in {\bf V}, e \in {\bf E} & \mbox{Variables}\\
\end{array} \]

Since the constraint matrix of the Lagrangian IP is totally unimodular, we can solve it in polynomial-time for every fixed $\lambda$.

\begin{lemma}
Let $\lambda_1$ be an arbitrary number in $(0,1)$. Assume $\sum_{i \in {\bf V}} x_i = k_1$ in the optimal solution to the Lagrangian IP. The set of $k_1$ vertices of ${\bf B}$ that correspond to $k_1$ nonzero variables of the Lagrangian IP covers the maximum number of edges of ${\bf B}$ among all subsets of $k_1$ vertices of ${\bf B}$.
\end{lemma}

\begin{proof}
Let ${\bf S}$ denote the subset of $k_1$ vertices of ${\bf B}$ obtained by solving the Lagrangian IP. For the purpose of contradiction, assume there exists a different subset ${\bf T}$ of $k_1$ vertices of ${\bf B}$ that covers mode edges than ${\bf S}$. Since the first term of the objective function will be the same for both subsets ($|{\bf S}| = |{\bf T}| = k_1$) and the second term of the objective function will be less for ${\bf T}$ than ${\bf S}$ (since ${\bf S}$ leaves mode edges uncovered, it has a higher penalty term), this will contradict ${\bf S}$ being the optimal solution to the Lagrangian IP.
\end{proof}

Let $A_n = \{a_1, a_2, \ldots\}$ be a sequence defined as $a_k = OPT_B (k) - OPT_B (k - 1)$. Notice that $B$ has the MNC property if and only if $A_n$ is a monotonically nonincreasing sequence, i.e., $a_{k+1} \leq a_{k} \forall k$.

\begin{lemma}
Let $\lambda$ be an arbitrary number in $(0,1)$ such that $\frac{1}{\lambda}$ is not integral. The number of nonzero $x_i$ variables in the optimal solution to the Lagrangian IP is equal to the number of elements of $A_n$ that are bigger than $\frac{1}{\lambda}$.
\end{lemma}

\begin{proof}
Assume the sequence $A_n$ has $k$ elements bigger than $\frac{1}{\lambda}$. For the purpose of contradiction, assume that in the optimal solution ${\bf S}$ to the Lagrangian IP the number of nonzero $x_i$ variables is $\alpha < k$. If we replace this solution with the optimal solution that has $k$ vertices then the first term of the objective function will increase by $k - \alpha$, however, the second term of the objective function will decline by more than $k - \alpha$ since each additional vertex covers at least $\frac{1}{\lambda}$ edges by construction. That will contradict the optimality of solution ${\bf S}$.

Similarly, assume in the optimal solution ${\bf S}$ to the Lagrangian IP the number of nonzero $x_i$ variables is $\alpha > k$. If we replace this solution with the optimal solution that has $k$ vertices then the first term of the objective function will decrease by $k - \alpha$, however, the second term of the objective function will increase by more than $\alpha - k$ since each removed vertex covers less than $\frac{1}{\lambda}$ edges by construction. That will contradict the optimality of solution ${\bf S}$.
\end{proof}

We will solve the Lagrangian IP for $\lambda$ values for which $\frac{1}{\lambda}$ is half-integral. For instance, $\lambda = \frac{1}{4.5}$, or $\lambda = \frac{1}{5.5}$, etc. We will make a binary search over $\lambda$ for values whose reciprocal is half-integral. When we do this binary search there are exactly two possibilities.

\begin{itemize}
\item We will either find a solution that covers exactly $t$ edges. In this case, we will report this solution.
\item We will either find $\lambda_1$ and $\lambda_2$ such that the solution for $\lambda_1$ covers less than $t$ edges, the solution for $\lambda_2$ covers more than $t$ edges end $\frac{1}{\lambda_2} = \frac{1}{\lambda_1} - 1$. Assume the optimal solution for $\lambda_1$ selects $k_1$ vertices and covers $t_1$ edges. Then the optimal solution to the PVCB problem will have $\lceil(k_1 + \frac{(t - t_1)} {(\frac{1}{\lambda_2} - \frac{1}{2})})\rceil$ vertices. This is because $k_1$ vertices can cover at most $t_1$ edges, and each additional vertex increases the number of covered edges by $(\frac{1}{\lambda_2} - \frac{1}{2})$.
\end{itemize}

\begin{lemma}
\label{lem:MNCdoesnothold}
The MNC property does not hold on all bipartite graphs.
\end{lemma}

\begin{proof}Given a bipartite graph ${\bf B}$, recall that $OPT_B (k)$ denotes the maximum number of edges of ${\bf B}$ that can be covered by a subset of $k$ vertices.

\begin{figure}[ht]
\begin{center}
 \begin{tikzpicture}

 \node[circle,fill=black,draw] at (-1,1) (n1) {};
 \node at (-1, 1.35) (l1) {$v_1$};

 \node[circle,fill=black,draw] at (-1,0) (n2) {};
 \node at (-1, 0.35) (l2) {$v_2$};

 \node[circle,fill=black,draw] at (-1,-1) (n3) {};
 \node at (-1, -0.65) (l3) {$v_3$};

 \node[circle,fill=black,draw] at (1,1) (n4) {};
 \node at (1, 1.35) (l4) {$v_4$};

 \node[circle,fill=black,draw] at (1,0) (n5) {};
 \node at (1, 0.35) (l5) {$v_5$};

 \node[circle,fill=black,draw] at (1,-1) (n6) {};
 \node at (1, -0.65) (l6) {$v_6$};

\node[circle,fill=black,draw] at (-2, 1.25) (n7) {};
\node at (-2.4, 1.25) (l7) {$v_7$};

\node[circle,fill=black,draw] at (-2, 0.75) (n8) {};
\node at (-2.4, 0.75) (l8) {$v_8$};

\node[circle,fill=black,draw] at (-2, 0.25) (n9) {};
\node at (-2.4, 0.25) (l9) {$v_9$};

\node[circle,fill=black,draw] at (-2, -0.25) (n10) {};
\node at (-2.4, -0.25) (l10) {$v_{10}$};

\node[circle,fill=black,draw] at (-2, -0.75) (n11) {};
\node at (-2.4, -0.75) (l11) {$v_{11}$};

\node[circle,fill=black,draw] at (-2, -1.25) (n12) {};
\node at (-2.4, -1.25) (l12) {$v_{12}$};

\node[circle,fill=black,draw] at (2, 1.5) (n13) {};
\node at (2.4, 1.5) (l13) {$v_{13}$};

\node[circle,fill=black,draw] at (2, 1) (n14) {};
\node at (2.4, 1) (l14) {$v_{14}$};

\node[circle,fill=black,draw] at (2, 0.5) (n15) {};
\node at (2.4, 0.5) (l15) {$v_{15}$};

\node[circle,fill=black,draw] at (2, 0) (n16) {};
\node at (2.4, 0) (l16) {$v_{16}$};

\draw (n1)--(n4);
\draw (n1)--(n5);
\draw (n1)--(n6);

\draw (n2)--(n4);
\draw (n2)--(n5);
\draw (n2)--(n6);

\draw (n3)--(n4);
\draw (n3)--(n5);
\draw (n3)--(n6);

\draw (n1)--(n7);
\draw (n1)--(n8);

\draw (n2)--(n9);
\draw (n2)--(n10);

\draw (n3)--(n11);
\draw (n3)--(n12);

\draw (n4)--(n13);
\draw (n4)--(n14);
\draw (n4)--(n15);

\draw (n5)--(n16);

 \end{tikzpicture}

\end{center}
\label{fig:MNCCounter}
\caption{An example bipartite graph $B$ for which the MNC property does not hold.}
  \end{figure}
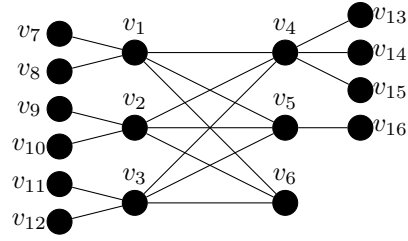

In the bipartite graph {\bf B} given in Figure \ref{fig:MNCCounter}, $OPT_B (1) = 6$. Notice that $\{v_4\}$ covers $6$ edges. $OPT_B(2) = 10$. There are several subsets of two vertices that cover $10$ edges. For instance, $\{v_4,v_5\}$, or $\{v_1, v_2\}$. $OPT_B (3) = 15$. There is only one subset of three vertices that covers $15$ edges and this subset is $\{v_1, v_2, v_3\}$.

Since $OPT_B (3) - OPT_B (2) \not\leq OPT_B (2) - OPT_B (2)$ for $B$, the MNC property does not hold on all bipartite graphs as stated by Lemma \ref{lem:MNCdoesnothold}.

\end{proof}

\subsection{MNC Property on Trees}

In this section, we prove that the MNC property holds on trees. This result is stated below as Lemma \ref{lem:MNCUT}. Notice that Lemma \ref{lem:MNCUT} and Theorem \ref{thm:PifMNPholds} imply that the PVCT problem is in {\bf P}. This result is presented below as Theorem \ref{thm:PVCTisinP}. In this section, we use the notation $OPT_T(k)$ to denote the maximum number of edges of a given tree ${\bf T}$ that can be covered by $k$ vertices.

\begin{lemma}
\label{lem:MNCUT}
MNC property holds on trees.
\end{lemma}

\begin{proof}
For the purpose of contradiction, assume that the MNC property does not hold for some tree {\bf T}, i.e., we have $OPT_T(k + 1) - OPT_T(k) > OPT_T(k) - OPT_T(k - 1)$ for some integer $k$. Without loss of generality, let $k$ be the smallest such integer.

\begin{lemma}
\label{lem:str}
Let {\bf A} be an optimal solution to the PVCT problem with $k + 1$ vertices. Every subset of {\bf A} with $k - 1$ vertices cover at most $OPT_T(k - 1) - 3$ edges.
\end{lemma}

\begin{proof}
Let {\bf B} be a subset of {\bf A} with $k - 1$ vertices and let it be a subset that covers the maximum number of edges among all such subsets. Let $|{\bf B}|$ and $|{\bf A}|$ denote the number of edges covered by ${\bf B}$ and ${\bf A}$ respectively. Notice that $|{\bf A}| = OPT_T(k + 1)$ by definition. The lemma states that $|{\bf B}| \leq OPT_T(k - 1) - 3$.

Let $\{i,j\} = {\bf A} - {\bf B}$. And let $\alpha = OPT_T(k) - OPT_T(k - 1)$, and let $\alpha + l = OPT_T(k + 1) - OPT_T(k)$. Notice that $l$ is a positive integer since the MNC property does not hold by assumption.

Notice that $\max\{|B \cup \{i\}|, |B \cup \{j\}|\} \geq \lceil \frac{|{\bf B}| + |{\bf A}|}{2}\rceil$. If $|{\bf B}| \geq OPT_T(k - 1) - 2$ then we will have $\max\{|B \cup \{i\}|, |B \cup \{j\}|\} > OPT_T(k - 1) + \alpha = OPT_T(k)$, which will be a contradiction. Therefore, the lemma holds.
\end{proof}

\begin{lemma}
\label{lem:str2}
Let {\bf X} be an optimal solution to PVCT problem with $k - 1$ vertices. Every superset of ${\bf X}$ with $k + 1$ vertices cover at most $OPT_T(k + 1) - 2$ edges.
\end{lemma}

\begin{proof}
Let ${\bf Y}$ be a superset of ${\bf X}$ with $k + 1$ vertices and let it be a superset that covers the maximum number of edges among all such supersets. Notice that $|{\bf X}| = OPT_T(k - 1)$ by definition. The lemma states that $|{\bf Y}| \leq OPT_T(k + 1)| - 2$.

Let $\{i,j\} = {\bf Y} - {\bf X}$. And let $\alpha = OPT_T(k) - OPT_T(k - 1)$, and let $\alpha + l = OPT_T(k + 1) - OPT_T(k)$. Notice that $l$ is a positive integer since the MNC property does not hold by assumption.

Notice that $\max\{|X \cup \{i\}|, |X \cup \{j\}|\} \geq \lceil \frac{|{\bf X}| + |{\bf Y}|}{2}\rceil$. If $|{\bf Y}| \geq OPT_T(k + 1) - 1$ then we will have $\max\{|X \cup \{i\}|, |X \cup \{j\}|\} > OPT_T(k - 1) + \alpha = OPT_T(k)$, which will be a contradiction. Therefore, the lemma holds.
\end{proof}

Let {\bf A} be an optimal solution to the PVCT problem with $k + 1$ vertices, and let {\bf X} be an optimal solution to the PVCT problem with $k - 1$ vertices. Lemma \ref{lem:str} implies that ${\bf X} \not\subset {\bf A}$, i.e., we have both ${\bf A} - {\bf X}$ and ${\bf X} - {\bf A}$ nonempty. Since ${\bf A}$ is composed of $k + 1$ vertices and ${\bf X}$ is composed of $k - 1$ vertices, ${\bf A} - {\bf X}$ has $2$ more elements than ${\bf X} - {\bf A}$. Let $i$ and $j$ be $2$ arbitrary elements of ${\bf A} - {\bf X}$. Notice that $|{\bf A}| - |{\bf A} - \{i,j\}| \geq |{\bf X} \cup \{i, j\}| - |{\bf X}| + 5$ due to Lemma \ref{lem:str} and Lemma \ref{lem:str2}. In other words, if we append the vertices $i$ and $j$ to the set ${\bf A} - \{i,j\}$, the increase in the number of edges covered will be at least $5$ more than the increase in the number of edges covered if $i$ and $j$ are appended to ${\bf X}$. Therefore, there are at least $5$ edges between the sets $\{i,j\}$ and ${\bf X} - {\bf A}$. Therefore, either $i$ or $j$ is incident to at least $3$ vertices of ${\bf X} - {\bf A}$. Without loss of generality, let $i$ be that vertex. Since a vertex of a tree can be the child of exactly one vertex, $i$ will be parenting at least $2$ vertices of ${\bf X} - {\bf A}$. Since this is true for every pair $\{i,j\}$ of vertices of ${\bf A} - {\bf X}$, at least all but one vertex of ${\bf A} - {\bf X}$ is parenting some vertex of ${\bf X} - {\bf A}$. This is a contradiction since ${\bf A} - {\bf X}$ has $2$ more vertices than ${\bf X} - {\bf A}$. Therefore, the lemma holds.
\end{proof}

\begin{theorem}
\label{thm:PVCTisinP}
The PVCT problem is in {\bf P}.
\end{theorem}

\begin{proof}Theorem directly follows from Lemma \ref{lem:MNCUT} and Theorem \ref{thm:PifMNPholds}.\end{proof}

\subsection{MNC Property on Degree Bounded Bipartite Graphs}

In this section, we show that the MNC property holds on bipartite graphs, if the degree of each vertex is at most $3$. Therefore, the PVC problem is polynomial-time solvable on the set of bipartite graphs, where the degree of each vertex is at most $3$. This result is stated below as Theorem \ref{thm:MNCholdsdegree3}.

\begin{theorem}
\label{thm:MNCholdsdegree3}
The PVC problem is in {\bf P} on the set of bipartite graphs, where the degree of each vertex is at most $3$.
\end{theorem}

\begin{proof}Notice that Lemma \ref{lem:str} and Lemma \ref{lem:str2} holds not only for trees but all bipartite graphs, since in the proofs of Lemma \ref{lem:str} and Lemma \ref{lem:str2} we have not used the fact that the underlying bipartite graph is a tree.

Let {\bf A} be an optimal solution to the PVCB problem with $k + 1$ vertices, and let {\bf X} be an optimal solution to the PVCB problem with $k - 1$ vertices. Lemma \ref{lem:str} implies that ${\bf X} \not\subset {\bf A}$, i.e., we have both ${\bf A} - {\bf X}$ and ${\bf X} - {\bf A}$ nonempty. Since ${\bf A}$ is composed of $k + 1$ vertices and ${\bf X}$ is composed of $k - 1$ vertices, ${\bf A} - {\bf X}$ has $2$ more elements than ${\bf X} - {\bf A}$.

Let $i$ and $j$ be $2$ arbitrary elements of ${\bf A} - {\bf X}$. Notice that $|{\bf A}| - |{\bf A} - \{i,j\}| \geq |{\bf X} \cup \{i, j\}| - |{\bf X}| + 5$ due to Lemma \ref{lem:str} and Lemma \ref{lem:str2}. In other words, if we append the vertices $i$ and $j$ to the set ${\bf A} - \{i,j\}$, the increase in the number of edges covered will be at least $5$ more than the increase in the number of edges covered if $i$ and $j$ are appended to ${\bf X}$. Therefore, there are at least $5$ edges between the sets $\{i,j\}$ and ${\bf X} - {\bf A}$. Therefore, either $i$ or $j$ is incident to at least $3$ vertices of ${\bf X} - {\bf A}$. Since any two elements of ${\bf A} - {\bf X}$ is incident to $5$ elements of ${\bf X} - {\bf A}$, the number of edges with one endpoint in ${\bf A} - {\bf X}$, one endpoint in ${\bf X} - {\bf A}$ is at least $3 \cdot \left(|{\bf A} - {\bf X}| - 1\right) + 2 = 3 \cdot |{\bf A} - {\bf X}| - 1$.  Since $|{\bf A} - {\bf X}| = |{\bf X} - {\bf A}| + 2$, we have the number of edges with one endpoint in ${\bf A} - {\bf X}$, one endpoint in ${\bf X} - {\bf A}$ is at least $3 \cdot |{\bf X} - {\bf A}| + 5$. Therefore, the vertices in ${\bf X} - {\bf A}$ are incident to at least $3 \cdot |{\bf X} - {\bf A}| + 5$ edges. This contradicts with the fact that the degree of every vertex of ${\bf B}$ is at most $3$.

\end{proof}

\subsection{MNC Property for Vertex Weighted Trees}

We proved Theorem \ref{thm:PifMNPholds}, for the PVC problem on bipartite graphs where neither the vertices nor the edges are weighted. By proper definition of the Marginally Nonincreasing Coverage property on vertex weighted graphs, Theorem \ref{thm:PifMNPholds} can be extended for the vertex weighted graph. For a bipartite graph ${\bf B}$, whose vertices are weighted, the analogues definition of the MNC property is as follows: Let $OPT_B(k)$ denote the maximum number of edges that can be covered by any subset ${\bf S}$ of vertices of {\bf B} such that the sum of the weights of the vertices of ${\bf S}$ is at most $k$. The the MNC property holds on ${\bf B}$ if $OPT_B(k+2) - OPT_B(k+1) \leq OPT_B(k+1) - OPT_B(k)$ for all integers $k$.

Though Theorem \ref{thm:PifMNPholds} can be generalized for the bipartite graphs with weighted vertices, this does not imply polynomial-time solvability for vertex weighted trees. This is because the MNC property does not hold for vertex weighted trees. Figure \ref{fig:counterexample} shows a counterexample where the weight of each vertex is in $\{1,2\}$.

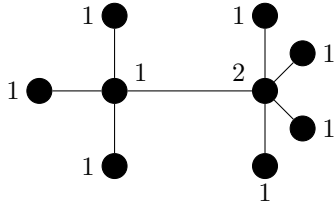
\begin{figure}[htb]
\begin{center}
\begin{tikzpicture}

\node [fill=black, circle, draw] at (-1, 1) (n1) {};
\node at (-1.35, 1) (l1) {$1$};

\node [fill=black, circle, draw] at (-2, 0) (n2) {};
\node at (-2.35, 0) (l2) {$1$};

\node [fill=black, circle, draw] at (-1, -1) (n3) {};
\node at (-1.35, -1) (l3) {$1$};

\node [fill=black, circle, draw] at (-1, 0) (n4) {};
\node at (-0.65, 0.25) (l4) {$1$};

\node [fill=black, circle, draw] at (1, 0) (n5) {};
\node at (0.65, 0.25) (l5) {$2$};

\node [fill=black, circle, draw] at (1, 1) (n6) {};
\node at (0.65, 1) (l6) {$1$};

\node [fill=black, circle, draw] at (1.5, 0.5) (n7) {};
\node at (1.85, 0.5) (l7) {$1$};

\node [fill=black, circle, draw] at (1.5, -0.5) (n8) {};
\node at (1.85, -0.5) (l8) {$1$};

\node [fill=black, circle, draw] at (1, -1) (n9) {};
\node at (1, -1.35) (l9) {$1$};

\draw (n1)--(n4);
\draw (n2)--(n4);
\draw (n3)--(n4);
\draw (n4)--(n5);
\draw (n6)--(n5);
\draw (n7)--(n5);
\draw (n8)--(n5);
\draw (n9)--(n5);
\end{tikzpicture}
\end{center}
\label{fig:counterexample}
\caption{Counterexample showing that the MNC property does not hold for the Partial Vertex Cover problem on vertex-weighted trees, even if the vertex weights are restricted to be in $\{1,2\}$.}
\end{figure}

Consider the example graph given in Figure \ref{fig:counterexample}. In this example, we have a tree with $9$ vertices. The weights of the vertices are written right next to them. So, the weight of $8$ of the vertices is $1$, and we have only $1$ vertex with a weight of $2$.

Let $OPT_T(k)$ denote the maximum number of edges that can be covered by a subset ${\bf S}$ of the vertices of the tree such that $\sum_{i \in {\bf S}} w_i \leq k$. In the example graph, it is easy to check $OPT_T(1) = 4$, $OPT_T(2) = 5$, and $OPT_T(3) = 8$.

Since $OPT_T(3) - OPT_T(2) \not \leq OPT_T(2) - OPT_T(1)$, MNC property does not hold for the Partial Vertex Cover problem on vertex-weighted trees, even if the vertex weights are in $\{1,2\}$.

This observation is stated below as Lemma \ref{lem:MNCWT}.

\begin{lemma}
\label{lem:MNCWT}
MNC property does not hold for the Partial Vertex Cover problem on trees when the edges are unweighted, and the vertex weights are in $\{1,2\}$.
\end{lemma}

\section{Conclusion}
\label{sec:conc}
In this paper, we studied the Partial Vertex Cover problem on bipartite graphs and trees. We proved that the PVC problem is {\bf NP-hard} on bipartite graphs by giving a reduction from the CLIQUE problem. We then proved that the problem is polynomial-time solvable for the set of bipartite graphs for which the Marginally Nonincreasing Coverage property holds. We also proved that the PVC problem is in {\bf P} on trees by showing that the MNC property holds on trees. We also proved the MNC property holds on the set of bipartite graphs, where the degree of each vertex is at most $3$. Therefore, proved that the PVCB problem is in {\bf P} on that set of bipartite graphs. We also analyzed the vertex-weighted trees and showed that the MNC property does not hold even on vertex-weighted trees, even if the vertex weights are in $\{1,2\}$.

From our perspective, the following lines of research appear promising:

\begin{itemize}
\item Determining whether the PVCB problem is ${\bf AXP-hard}$ or not.
\item Obtaining an $\alpha$-approximation algorithm where $\alpha \leq \frac{4}{3}$.
\end{itemize}

\bibliographystyle{plain}
\bibliography{bugrarefs,orefs,mrefs,myrefs,mytrefs,negcyclerefs,proofrefs,sprefs,graphrefs,crefs,satrefs,symbrefs,constraintrefs,modcheckrefs,diffsymbrefs,certrefs,classrefs}

\end{document}